\documentclass[12pt]{article}

\usepackage{amsmath}
\usepackage{amssymb}
\usepackage{amsfonts}
\usepackage{latexsym}
\usepackage{color}

\catcode `\@=11 \@addtoreset{equation}{section}

\catcode `\@=12

\newcommand{\be}{\begin{equation}}
\newcommand{\en}{\end{equation}}
\newcommand{\bea}{\begin{eqnarray}}
\newcommand{\ena}{\end{eqnarray}}
\newcommand{\beano}{\begin{eqnarray*}}
\newcommand{\enano}{\end{eqnarray*}}
\newcommand{\bee}{\begin{enumerate}}
\newcommand{\ene}{\end{enumerate}}

\newcommand{\N}{\mathfrak N}

\newcommand{\mc}{\mathcal}

\newcommand{\D}{{\mc D}}

\newcommand{\F}{{\cal F}}

\newcommand{\Lc}{{\cal L}}

\newcommand{\1}{1 \!\! 1}

\newcommand{\Hil}{\mc H}

\newtheorem{thm}{Theorem}

\newenvironment{proof}{\noindent {\bf Proof --}}{\hfill$\square$ \vspace{3mm}\endtrivlist}

\catcode `\@=11 \@addtoreset{equation}{section}
\catcode `\@=12

\begin{document}

\thispagestyle{empty}

\vspace*{2cm}

\begin{center}
{\Large \bf  (Regular) pseudo-bosons versus bosons}\\[10mm]

{\large F. Bagarello}\\
  Dipartimento di Metodi e Modelli Matematici,
Facolt\`a di Ingegneria,\\ Universit\`a di Palermo, I-90128  Palermo, Italy\\
e-mail: bagarell@unipa.it\\ Home page:
www.unipa.it/$^\sim$bagarell\\

\vspace{3mm}

\end{center}

\vspace*{2cm}

\begin{abstract}
\noindent We discuss in which sense the so-called {\em regular
pseudo-bosons}, recently introduced by Trifonov and analyzed in some details by the author, are related to ordinary bosons. We repeat the same analysis also for {\em pseudo-bosons}, and we analyze the role played by certain intertwining operators, which may be bounded or not.

\end{abstract}

\vspace{2cm}


\vfill


\newpage

\section{Introduction}

In a series of recent papers \cite{bagpb1,bagpb2,bagpb3,bagcal}, we have investigated some mathematical aspects of the
so-called {\em pseudo-bosons} (PB),  originally introduced by Trifonov
in \cite{tri}. They arise from the canonical commutation relation
$[a,a^\dagger]=\1$ upon replacing $a^\dagger$ by another (unbounded)
operator $b$ not (in general) related to $a$: $[a,b]=\1$. We have
shown that, under suitable assumptions, $N=ba$ and $\N=N^\dagger=a^\dagger b^\dagger$ can be both
diagonalized, and that their spectra coincide with the set of
natural numbers (including 0), ${\Bbb N}_0$. However the sets of
related eigenvectors are not orthonormal (o.n.) bases but,
nevertheless, they are automatically {\em biorthogonal}. In most of the
examples considered so far, they are bases of the Hilbert space of the system,
$\Hil$, and, in some cases, they turn out to be {\em Riesz bases}.

In \cite{bagpb4} and \cite{abg} some physical examples arising from concrete models in quantum mechanics have been discussed. These examples suggested to introduce the difference between {\em regular pseudo-bosons} (RPB) and PB: the RPB, see Section II, arise when the two sets of eigenvectors of $N$ and $\N$ are mapped one into the other by a bounded operator with bounded inverse. If this operator is unbounded, then we have to do with PB. PB have also been considered by other authors recently, see \cite{jan} for instance, without calling them in this way.  These PB have been shown to have to do with the so-called pseudo-hermitian quantum mechanics, which in recent years have became more and more appealing since it considers the possibility of having non self-adjoint hamiltonians with real spectra, showing that this possibility is related to some commutativity conditions between the hamiltonian itself and the parity and the time reversal operators, \cite{ben}. The same feature, more from a mathematical side, has been analyzed for instance in \cite{mosta,mosta2}. Of course, these references should be considered just as a starting point for a deeper analysis.

In this paper we consider the relation between PB, RPB, and ordinary bosons, proving two similar theorems, one for  PB and the other for RPB. More in details:  in the next section we
introduce and discuss some features of $d$-dimensional PB.  In Sections III we prove our main theorem for RPB, while Section IV contains an analogous result for PB, together with some physical examples; we will see that techniques of unbounded operators are the natural tools in that case. We give our conclusions in Section V.

\section{$d$-dimensional PB and RPB}

In this section we will construct a $d$-dimensional ($d$-D) version  of
what originally proposed in \cite{bagpb1}, to which we refer for
further comments on the 1-D situation.

Let $\Hil$ be a given Hilbert space with scalar product
$\left<.,.\right>$ and related norm $\|.\|$. We introduce $d$
pairs of operators, $a_j$ and $b_j$, $j=1,2,\ldots,d$, acting on $\Hil$ and
satisfying the following commutation rules \be [a_j,b_j]=\1,
\label{21} \en where $j=1,2,\ldots,d$,  all the other commutators being trivial. Of course, they collapse to the CCR's for $d$
independent modes if $b_j=a^\dagger_j$, $j=1,2,\ldots,d$. It is well known
that $a_j$ and $b_j$ are unbounded operators, so they cannot be
defined on all of $\Hil$. Following \cite{bagpb1}, and writing
$D^\infty(X):=\cap_{p\geq0}D(X^p)$ (the common  domain of all the powers of the
operator $X$), we consider the
following:

\vspace{2mm}

{\bf Assumption 1.--} there exists a non-zero
$\varphi_{\bf 0}\in\Hil$ such that $a_j\varphi_{\bf 0}=0$, $j=1,2,\ldots,d$,
and $\varphi_{\bf 0}\in D^\infty(b_1)\cap D^\infty(b_2)\cap\cdots\cap D^\infty(b_d)$.

{\bf Assumption 2.--} there exists a non-zero $\Psi_{\bf 0}\in\Hil$
such that $b_j^\dagger\Psi_{\bf 0}=0$, $j=1,2,\ldots,d$, and $\Psi_{\bf 0}\in
D^\infty(a_1^\dagger)\cap D^\infty(a_2^\dagger)\cap\cdots\cap D^\infty(a_d^\dagger)$.

\vspace{2mm}

Under these assumptions we can introduce the following vectors in
$\Hil$:

\be
\left\{
\begin{array}{ll}
\varphi_{\bf n}:=\varphi_{n_1,n_2,\ldots,n_d}=\frac{1}{\sqrt{n_1!n_2!\cdots n_d!}}\,b_1^{n_1}\,b_2^{n_2}\cdots b_d^{n_d}\,\varphi_{\bf 0}\\
\Psi_{\bf n}:=\Psi_{n_1,n_2,\ldots,n_d}=\frac{1}{\sqrt{n_1!n_2!\cdots n_d!}}\,{a_1^\dagger}^{n_1}\,{a_2^\dagger}^{n_2}\cdots {a_d^\dagger}^{n_d}\,\Psi_{\bf 0},
\end{array}
\right.\label{22}\en $n_j=0, 1, 2,\ldots$, for all $j=1,2,\ldots,d$. Let us now define the unbounded
operators $N_j:=b_ja_j$ and $\N_j:=N_j^\dagger=a_j^\dagger
b_j^\dagger$, $j=1,2,\ldots,d$.  It is possible to check that
$\varphi_{\bf n}$ belongs to the domain of $N_j$, $D(N_j)$, and that
$\Psi_{\bf n}\in D(\N_j)$, for all possible $\bf n$. Moreover,
\be N_j\varphi_{\bf n}=n_j\varphi_{\bf n},  \quad \N_j\Psi_{\bf n}=n_j\Psi_{\bf n}.
\label{23}\en

Under the above assumptions, and if we chose the normalization of
$\Psi_{\bf 0}$ and $\varphi_{\bf 0}$ in such a way that
$\left<\Psi_{\bf 0},\varphi_{\bf 0}\right>=1$, we find that \be
\left<\Psi_{\bf n},\varphi_{\bf m}\right>=\delta_{\bf n,m}=\prod_{j=1}^d \delta_{n_j,m_j}. \label{27}\en This means that the sets
$\F_\Psi=\{\Psi_{\bf n}\}$ and
$\F_\varphi=\{\varphi_{\bf n}\}$ are {\em biorthogonal} and,
because of this, the vectors of each set are linearly independent.
If we now call $\D_\varphi$ and $\D_\Psi$ respectively the linear
span of  $\F_\varphi$ and $\F_\Psi$, and $\Hil_\varphi$ and
$\Hil_\Psi$ their closures, then \be f=\sum_{\bf n}
\left<\Psi_{\bf n},f\right>\,\varphi_{\bf n}, \quad \forall
f\in\Hil_\varphi,\qquad  h=\sum_{\bf n}
\left<\varphi_{\bf n},h\right>\,\Psi_{\bf n}, \quad \forall
h\in\Hil_\Psi. \label{210}\en What is not in general ensured is
that the Hilbert spaces introduced so far all coincide, i.e. that
$\Hil_\varphi=\Hil_\Psi=\Hil$. Indeed, we can only state that
$\Hil_\varphi\subseteq\Hil$ and $\Hil_\Psi\subseteq\Hil$. However,
motivated by the examples discussed in the literature,  we make the

\vspace{2mm}

{\bf Assumption 3.--} The above Hilbert spaces all coincide:
$\Hil_\varphi=\Hil_\Psi=\Hil$,

\vspace{2mm}

which was introduced in \cite{bagpb1}. This means, in particular,
that both $\F_\varphi$ and $\F_\Psi$ are bases of $\Hil$, so that the following resolutions of the identity, written in bra-ket notation, hold:
\be
\sum_{\bf n}|\varphi_{\bf n}\left>\right<\Psi_{\bf n}|=\sum_{\bf n}|\Psi_{\bf n}\left>\right<\varphi_{\bf n}|=\1.
\label{210b}\en
Let us
now introduce the operators $S_\varphi$ and $S_\Psi$ via their
action respectively on  $\F_\Psi$ and $\F_\varphi$: \be
S_\varphi\Psi_{\bf n}=\varphi_{\bf n},\qquad
S_\Psi\varphi_{\bf n}=\Psi_{\bf n}, \label{213}\en for all $\bf n$, which in particular imply that
$\Psi_{\bf n}=(S_\Psi\,S_\varphi)\Psi_{\bf n}$ and
$\varphi_{\bf n}=(S_\varphi \,S_\Psi)\varphi_{\bf n}$, for all
$\bf n$. Hence \be S_\Psi\,S_\varphi=S_\varphi\,S_\Psi=\1 \quad
\Rightarrow \quad S_\Psi=S_\varphi^{-1}. \label{214}\en In other
words, both $S_\Psi$ and $S_\varphi$ are invertible and one is the
inverse of the other. Furthermore, we can also check that they are
both positive, well defined and symmetric, \cite{bagpb1}. Moreover, at
least formally, it is possible to write these operators  as \be S_\varphi=\sum_{\bf n}\,
|\varphi_{\bf n}><\varphi_{\bf n}|,\qquad S_\Psi=\sum_{\bf n}
\,|\Psi_{\bf n}><\Psi_{\bf n}|. \label{212}\en
 These expressions are
only formal, at this stage, since the series may or may not converge in
the uniform topology and the operators $S_\varphi$ and $S_\Psi$ could be unbounded.
Indeed we know,  \cite{you}, that two biorthogonal bases are related by a bounded operator, with bounded inverse, if and only if they are Riesz bases\footnote{Recall that a set of vectors $\phi_1, \phi_2 , \phi_3 , \; \ldots \; ,$ is a Riesz basis of a Hilbert space $\mathcal H$, if there exists a bounded operator $V$, with bounded inverse, on $\mathcal H$, and an orthonormal basis of $\Hil$,  $\varphi_1, \varphi_2 , \varphi_3 , \; \ldots \; ,$ such that $\phi_j=V\varphi_j$, for all $j=1, 2, 3,\ldots$}. This is why in \cite{bagpb1} we have also considered

\vspace{2mm}

{\bf Assumption 4.--} $\F_\varphi$ and $\F_\Psi$ are Bessel sequences. In other words, there exist two positive constants $A_\varphi,A_\Psi>0$ such that, for all $f\in\Hil$,
\be
\sum_{\bf n}\,|\left<\varphi_{\bf n},f\right>|^2\leq A_\varphi\,\|f\|^2,\qquad \sum_{\bf n}\,|\left<\Psi_{\bf n},f\right>|^2\leq A_\Psi\,\|f\|^2.
\label{215}\en

\vspace{3mm}
\noindent
This assumption is equivalent to require that $\F_\varphi$ and $\F_\Psi$ are both Riesz bases, and implies  that $S_\varphi$ and $S_\Psi$ are bounded operators:
$\|S_\varphi\|\leq A_\varphi$, $\|S_\Psi\|\leq A_\Psi$. Moreover
$
\frac{1}{A_\Psi}\,\1\leq S_\varphi \leq A_\varphi\,\1,$ and  $\frac{1}{A_\varphi}\,\1\leq S_\Psi \leq A_\Psi\,\1.
$
Hence the domains of $S_\varphi$ and $S_\Psi$ can be taken to be all of $\Hil$. While Assumptions 1, 2 and 3 are quite often satisfied, as the examples contained in our previous papers and  in the recent review \cite{bagrev} show,  it is quite difficult to find {\bf physical} examples satisfying also Assumption 4. On the other hand, it is rather easy to find {\bf mathematical} examples satisfying all the assumptions, see Section II.1 below. Hence, as announced, we introduce the following difference: we call {\em pseudo-bosons} (PB) those {\em excitations} satisfying  the first three assumptions, while, if  Assumption 4 is also satisfied, these will be called {\em regular pseudo-bosons} (RPB). Clearly, RPB are PB, but the converse is false, in general.

Generalizing what already discussed in \cite{bagpb1,abg},  these
$d$-dimensional pseudo-bosons give rise to interesting
intertwining relations among non self-adjoint operators, see also
\cite{bagpb3} and references therein. In particular it is easy to
check that \be S_\Psi\,N_j=\N_jS_\Psi \quad \mbox{ and }\quad
N_j\,S_\varphi=S_\varphi\,\N_j, \label{219}\en $j=1,2,\ldots,d$. This is
related to the fact that the spectra of, say, $N_1$ and $\N_1$
coincide and that their eigenvectors are related by the operators
$S_\varphi$ and $S_\Psi$, see equations (\ref{23}) and (\ref{213}), in agreement with the literature on
intertwining operators, \cite{intop,bag1}, and on pseudo-Hermitian
quantum mechanics, see \cite{ben,mosta,mosta2} and references therein.

\subsection{Construction of RPB}

We will show here that each Riesz basis produces some RPB. Let $\F_\varphi:=\{\varphi_{\bf n}\}$ be a Riesz basis of $\Hil$ with bounds $A$ and $B$, $0<A\leq B<\infty$. The associated frame operator $S:=\sum_{\bf n}\,|\varphi_{\bf n}><\varphi_{\bf n}|$ is bounded, positive and admits a bounded inverse. Also, the set $\F_{\hat\varphi}:=\{\hat\varphi_{\bf n}:=S^{-1/2}\varphi_{\bf n}\}$ is an o.n. basis of $\Hil$. Therefore we can define $d$ lowering operators $a_{j,\hat\varphi}$ on  $\F_{\hat\varphi}$ as $a_{j,\hat\varphi}\hat\varphi_{\bf n}=\sqrt{n_j}\,\hat\varphi_{\bf n_{j-}}$, and their adjoints, $a_{j,\hat\varphi}^\dagger$, as $a_{j,\hat\varphi}^\dagger\hat\varphi_{\bf n}=\sqrt{n_j+1}\,\hat\varphi_{\bf n_{j+}}$. Here ${\bf n}_{j-}=(n_1,\ldots,n_j-1,\ldots,n_d)$ and ${\bf n}_{j+}=(n_1,\ldots,n_j+1,\ldots,n_d)$. Hence $[a_{j,\hat\varphi},a_{k,\hat\varphi}^\dagger]=\delta_{j,k}\,\1$. If we now define $a_j:=S^{1/2}\,a_{j,\hat\varphi}\,S^{-1/2}$, this acts on the Riesz basis $\F_\varphi$ as a lowering operator. However, since $\F_\varphi$ is not an o.n. basis in general, $a_j^\dagger$ is not a raising operator,  so that $[a_j,a_k^\dagger]\neq\delta_{j,k}\,\1$. However, if we now define the operator $b_j:=S^{1/2}\,a_{j,\hat\varphi}^\dagger\,S^{-1/2}$, it is clear that in general  $b_j\neq a_j^\dagger$, and $b_j$ acts on $\varphi_{\bf n}$ as a raising operator: $b_j\,\varphi_{\bf n}=\sqrt{n_j+1}\,\varphi_{{\bf n}_{j+}}$, for all $\bf n$. Then we have $[a_j,b_k]=\delta_{j,k}\,\1$. So we have constructed two sets of operators satisfying (\ref{21}) and which are not related by a simple conjugation. This is not the end of the story. Indeed:
\begin{enumerate}
\item Assumption 1 is verified since $\varphi_{\bf 0}$ is annihilated by $a_j$ and belongs to the domain of all the powers of $b_j$.
\item As for Assumption 2, it is enough to define $\Psi_{\bf 0}=S^{-1}\,\varphi_{\bf 0}$. With this definition $b_j^\dagger\,\Psi_{\bf 0}=0$ and $\Psi_{\bf 0}$ belongs to the domain of all the powers of $a_j^\dagger$.
\item Since $\F_\varphi$ is a Riesz basis of $\Hil$ by assumption, then $\Hil_\varphi=\Hil$. Notice now that the vectors $\Psi_{\bf n}$  can be written as $\Psi_{\bf n}=S^{-1}\,\varphi_{\bf n}$, for all $\bf n$. Hence $\F_\Psi$ is in duality with $\F_\varphi$ and therefore is a Riesz basis of $\Hil$ as well. Hence $\Hil_\Psi=\Hil$. This proves Assumption 3.
\item As for Assumption 4, this is equivalent to the hypothesis originally assumed here, i.e. that $\F_\varphi$ is a Riesz basis.
\end{enumerate}

Explicit examples arising from this general construction can be found in \cite{bagcal}.

\subsection{Coherent states}

As it is well known there exist several different, and not always
equivalent, ways to define {\em
coherent states}, \cite{book1,book2}. In this paper, following \cite{bagpb1}, we will
adopt the following definition: let $z_j$, $j=1,2,\ldots,d$ be $d$ complex variables, $z_j\in \D$ (some domain in $\Bbb{C}$), and let us introduce the following
operators: \be \left\{
\begin{array}{ll}
U_j(z_j)=e^{z_j\,b_j-\overline{z}_j\,a_j}=e^{-|z_j|^2/2}\,e^{z_j\,b_j}\,e^{-\overline{z}_j\,a_j},
\\
V_j(z_j)=e^{z_j\,a_j^\dagger-\overline{z}_j\,b_j^\dagger}=e^{-|z_j|^2/2}\,e^{z_j\,a_j^\dagger}
\,e^{-\overline{z}_j\,b_j^\dagger},
\end{array}
\right.\label{31}\en $j=1,2,\ldots,d$,  \be\left\{
\begin{array}{ll}
U(z_1,z_2,\ldots,z_d):=U_1(z_1)\,U_2(z_2)\,\cdots\,U_d(z_d),\\
V(z_1,z_2,\ldots,z_d):=V_1(z_1)\,V_2(z_2)\,\cdots\,V_d(z_d),\end{array}
\right. \label{31b}\en and the following
vectors: \be \varphi(z_1,z_2,\ldots,z_d)=U(z_1,z_2,\ldots,z_d)\varphi_{\bf 0},\qquad
\Psi(z_1,z_2,\ldots,z_d)=V(z_1,z_2,\ldots,z_d)\,\Psi_{\bf 0}. \label{32}\en \vspace{2mm}

{\bf Remarks:--} (1) Due to the commutation rules for the
operators $b_j$ and $a_j$, we  clearly have
$[U_j(z_j),U_k(z_k)]=[V_j(z_j),V_k(z_k)]=0$, for $j\neq k$.

(2) Since the operators $U$ and $V$ are, for generic $z_j$, unbounded, definition (\ref{32}) makes sense only if
$\varphi_{\bf 0}\in D(U)$ and $\Psi_{\bf 0}\in D(V)$, a condition which
will be assumed here. In \cite{bagpb1} it was proven that, for
instance, this is so when $\F_\varphi$ and $\F_\Psi$ are Riesz
bases.

(3) The set $\D$ could be,  in
principle, a proper subset of $\Bbb{C}$.

\vspace{2mm}

It is possible to write the vectors $\varphi(z_1,z_2,\ldots,z_d)$ and
$\Psi(z_1,z_2,\ldots,z_d)$ in terms of the vectors of $\F_\Psi$ and
$\F_\varphi$ as \be
\left\{
\begin{array}{ll}
\varphi(z_1,z_2,\ldots,z_d)=e^{-(|z_1|^2+|z_2|^2+\ldots+|z_d|^2)/2}\,\sum_{\bf n}\,\frac{z_1^{n_1}\,z_2^{n_2}\cdots z_d^{n_d}}{\sqrt{n_1!\,n_2!\ldots n_d!}}\,\varphi_{\bf n},\\
\\
\Psi(z_1,z_2,\ldots,z_d)=e^{-(|z_1|^2+|z_2|^2+\ldots+|z_d|^2)/2}\,\sum_{\bf n}\,\frac{z_1^{n_1}\,z_2^{n_2}\cdots z_d^{n_d}}{\sqrt{n_1!\,n_2!\ldots n_d!}}\,\Psi_{\bf n}.\end{array}
\right.
\label{33}\en

These vectors are called {\em coherent} since they are eigenstates
of the lowering operators. Indeed we can check that \be
a_j\varphi(z_1,z_2,\ldots,z_d)=z_j\varphi(z_1,z_2,\ldots,z_d), \qquad
b_j^\dagger\Psi(z_1,z_2,\ldots,z_d)=z_j\Psi(z_1,z_2,\ldots,z_d), \label{34}\en for
$j=1,2,\ldots,d$ and $z_j\in\D$. It is also a standard exercise, putting
$z_j=r_j\,e^{i\theta_j}$, to check that the following operator
equalities hold: \be
\left\{
\begin{array}{ll}
\frac{1}{\pi^d}\int_{\Bbb{C}}\,dz_1\int_{\Bbb{C}}\,dz_2\,\ldots\int_{\Bbb{C}}\,dz_d\,
|\varphi(z_1,z_2,\ldots,z_d)><\varphi(z_1,z_2,\ldots,z_d)|=S_\varphi, \\
\frac{1}{\pi^d}\int_{\Bbb{C}}\,dz_1\int_{\Bbb{C}}\,dz_2\,\ldots\int_{\Bbb{C}}\,dz_d\,
|\Psi(z_1,z_2,\ldots,z_d)><\Psi(z_1,z_2,\ldots,z_d)|=S_\Psi,\end{array}
\right. \label{35a}\en as well as
\be \frac{1}{\pi^d}\int_{\Bbb{C}}\,dz_1\int_{\Bbb{C}}\,dz_2\,\ldots\int_{\Bbb{C}}\,dz_d\,
|\varphi(z_1,z_2,\ldots,z_d)><\Psi(z_1,z_2,\ldots,z_d)|=\sum_{\bf n}|\varphi_{\bf n}\left>\right<\Psi_{\bf n}|=\1, \label{36a}\en which are
written in convenient bra-ket notation. It should be said that
these equalities are, most of the times, only formal results.
Indeed, extending an analogous result given in \cite{abg} for $d=2$, we can prove the following

\vspace{2mm}

\begin{thm}\label{thm1} Let $a_j$, $b_j$, $\F_\varphi$, $\F_\Psi$, $\varphi(z_1,z_2,\ldots,z_d)$ and $\Psi(z_1,z_2,\ldots,z_d)$ be as above. Let us assume that
(1) $\F_\varphi$, $\F_\Psi$ are Riesz bases; (2) $\F_\varphi$,
$\F_\Psi$ are biorthogonal. Then (\ref{36a}) holds true.

\end{thm}

Suppose therefore that
the above construction gives coherent states that do not satisfy a
resolution of the identity (see \cite{bagpb2} for an example). Then, since $\F_\varphi$ and
$\F_\Psi$ are automatically biorthogonal, they cannot be Riesz
bases (neither one of them)!

\section{RPB versus bosons}

In this section we will prove the following theorem, given in $d=1$ for simplicity, establishing a sort of equivalence between RPB and ordinary bosons. This equivalence is related to the existence of a bounded operator $T$ with bounded inverse and of a pair of conjugate operators $c$ and $c^\dagger$ satisfying the canonical commutation rule $[c,c^\dagger]=\1$, which are related with the original pair of operators $a$ and $b$. More in details we have:

\begin{thm}
Let $a$ and $b$ be two operators on $\Hil$ satisfying  $[a,b]=\1$, and for which Assumptions 1, 2, 3 and 4 of Section II are satisfied. Then an unbounded, densely defined, operator $c$ on $\Hil$ exists, together with a positive bounded operator $T$ with bounded inverse $T^{-1}$, such that $[c,c^\dagger]=\1$. Moreover
\be
a=T\,c\,T^{-1},\qquad b=T\,c^\dagger T^{-1}.
\label{41}\en
Viceversa, given an unbounded, densely defined, operator $c$ on $\Hil$ satisfying $[c,c^\dagger]=\1$ and a positive bounded operator $T$ with bounded inverse $T^{-1}$, two operators $a$ and $b$ can be introduced for which $[a,b]=\1$, and for which equations (\ref{41}) and Assumptions 1, 2, 3 and 4 of Section II are satisfied.

\end{thm}

\begin{proof}

To prove the first part of the theorem we first remind that, because of Assumption 4 of Section II, the operators $S_\varphi$ and $S_\Psi$ defined as in (\ref{212}), \be S_\varphi\,f=\sum_{n=0}^\infty\,
\left<\varphi_n,f\right>\varphi_{n},\qquad S_\Psi\,f=\sum_{n}
\,\left<\Psi_{n},f\right>\Psi_{n}, \label{42}\en
$f\in\Hil$, are well defined, bounded and positive (hence, self-adjoint). Also, $S_\varphi=S_\Psi^{-1}$ . These are standard results in the theory of Riesz bases, \cite{you,chri}. In particular, choosing the normalization constants in $\Psi_0$ and $\varphi_0$ in such a way that $\left<\Psi_0,\varphi_0\right>=1$, we know that  $\left<\Psi_n,\varphi_m\right>=\delta_{n,m}$ and, as a consequence,
\be
S_\varphi\,\Psi_m=\varphi_m,\qquad S_\Psi\,\varphi_m=\Psi_m,
\label{43}\en
for all $m\geq0$. Because of the properties of $S_\Psi$ and $S_\varphi$, their square roots surely exist and, for instance, $S_\varphi^{-1/2}=S_\Psi^{1/2}$. Hence we define the vectors $\hat\varphi_n=S_\varphi^{-1/2}\varphi_n$, $n\geq0$, and the related set $\F_{\hat\varphi}=\{\hat\varphi_n,\,n\geq0\}$. It is well known that $\F_{\hat\varphi}$ is an o.n. basis of $\Hil$, and it coincides with the o.n. basis we would construct introducing (apparently) new vectors  $\hat\Psi_n=S_\Psi^{-1/2}\Psi_n$, $n\geq0$, since it can be easily checked that, for all $n$, $\hat\Psi_n=\hat\varphi_n$.

On $\F_{\hat\varphi}$ we can define the ordinary bosonic lowering and raising operators:
\be
\left\{
\begin{array}{ll}
c\,\hat\varphi_n=\sqrt{n }\,\hat\varphi_{n-1},\\
c^\dagger\,\hat\varphi_n=\sqrt{n+1 }\,\hat\varphi_{n+1},\\\end{array}
\right.
\label{44}\en
with the convention that $c\,\hat\varphi_0=0$. Of course $[c,c^\dagger]=\1$. Recall now, \cite{bagpb1}, that  our working hypotheses also imply that $a\,\varphi_n=\sqrt{n }\,\varphi_{n-1}$ and
$b\,\varphi_n=\sqrt{n+1 }\,\varphi_{n+1}$, which can be rewritten as $S_\varphi^{-1/2}a\,S_\varphi^{1/2}\,\hat\varphi_n=\sqrt{n }\,\hat\varphi_{n-1},$ and
$S_\varphi^{-1/2}b\,S_\varphi^{1/2}\,\hat\varphi_n=\sqrt{n+1 }\,\hat\varphi_{n+1}$. Hence $a$, $b$ and $c$ are related as follows:
$$
c=S_\varphi^{-1/2}a\,S_\varphi^{1/2}, \qquad c^\dagger=S_\varphi^{-1/2}b\,S_\varphi^{1/2},
$$
which are exactly  equations (\ref{41}), identifying $T$ with $S_\varphi^{1/2}$.

\vspace{2mm}

Let us now prove the second part of the theorem. First of all, by means of $c$ and $c^\dagger$, we construct the o.n. basis $\F_{\hat\varphi}$ of $\Hil$, $\F_{\hat\varphi}=\left\{\hat\varphi_n=\frac{{c^\dagger}^n}{\sqrt{n! }}\,\hat\varphi_0\right\}$, where $c\,\hat\varphi_0=0$. Then, since both $T$ and $T^{-1}$ are bounded and, therefore, everywhere defined, we can introduce two new families of vectors: $\F_\varphi=\left\{\varphi_n=T\,\hat\varphi_n,\,n\geq0\right\}$ and $\F_\Psi=\left\{\Psi_n=T^{-1}\,\hat\varphi_n,\,n\geq0\right\}$. These two families are obviously biorthogonal, $\left<\Psi_n,\varphi_m\right>=\delta_{n,m}$, and they are both complete in $\Hil$: so they are two (in general different) bases of $\Hil$. We can now define on, say, $\F_\varphi$, two operators $a$ and $b$ which act as lowering and raising operators:
\be
\left\{
\begin{array}{ll}
a\,\varphi_n=\sqrt{n }\,\varphi_{n-1},\\
b\,\varphi_n=\sqrt{n+1 }\,\varphi_{n+1},\\
\end{array}
\right.
\label{45}\en
for all $n\geq0$. In particular the first equation implies that $a\varphi_0=0$. Incidentally we observe that $b^\dagger\neq a$, since $\F_{\varphi}$ is not, in general, an o.n. basis. Iterating the second equation in (\ref{45}), we deduce that $\varphi_n=\frac{b^n}{\sqrt{n! }}\,\varphi_0$, which gives an alternative expression for the vector $\varphi_n$ and, moreover, shows that $\varphi_0\in D^\infty(b)$. Hence Assumption 1 is satisfied.

Since $(T\,c^\dagger T^{-1})\varphi_n=T\,c^\dagger\hat\varphi_n=\sqrt{n+1 }\,T\hat\varphi_{n+1}=\sqrt{n+1 }\,\varphi_{n+1}$, and since $\F_\varphi$ is a basis of $\Hil$, we deduce that $b=T\,c^\dagger T^{-1}$. Analogously, we can prove that $a=T\,c\, T^{-1}$. It is now clear that $[a,b]=\1$ and that $a^\dagger=T^{-1}c^\dagger T$. To prove Assumption 2 we first notice that $b^\dagger \Psi_0=\left(T\,c^\dagger T^{-1}\right)^\dagger (T^{-1}\hat\varphi_0)= T^{-1} c\, \hat\varphi_0=0$. Moreover, since for all $n\geq0$
$$
a^\dagger\Psi_n=\left(T^{-1}c^\dagger T\right)T^{-1}\hat\varphi_n=T^{-1}c^\dagger \hat\varphi_n=\sqrt{n+1}\,\Psi_{n+1},
$$
by iteration we deduce that $\Psi_n=\frac{{a^\dagger}^n}{\sqrt{n! }}\,\Psi_0$, which means that $\Psi_0\in D^\infty(a^\dagger)$. This prove Assumption 2, while Assumption 3 follows from our previous claim on $\F_\varphi$ and $\F_\Psi$: they are both bases of $\Hil$. Finally, since they are obtained by the o.n. basis $\F_{\hat\varphi}$ by acting with the bounded operators $T$ or $T^{-1}$, they are also Riesz bases.

\end{proof}

\vspace{2mm}

{\bf Remarks:--} (1) The proof of the above theorem recall, at least in part, the construction given in Section II.1. This is not surprising since we are now dealing with Riesz bases. The difference will be evident in the next Section.

(2) Theorem 2 implies that the intertwining operators in (\ref{219}) for RPB are bounded, with bounded inverse.

\section{PB versus bosons}

In this section we will not assume that $T$ and $T^{-1}$ are bounded operators, and many domain problems will arise as a consequence. This will be related to the nature of the biorthogonal bases we work with, which will not be Riesz bases any longer. The relevance of this section, as widely  explained in \cite{bagrev} and references therein, follows from the fact that all the physical  examples seem to give rise to PB and not to RPB. From the mathematical side, we will formulate now a different theorem which is the analogue of the one proven in the previous section in this different settings and we will show that, even if part of that proof can be repeated here, most of the arguments should be changed to take care of unboundedness of the operators. As in the previous section, to simplify the proof and the notation, we fix $d=1$. Extension to $d>1$ is straightforward.

\begin{thm}
Let $a$ and $b$ be two operators on $\Hil$ satisfying  $[a,b]=\1$, and for which Assumptions 1, 2, and 3 (but not 4)  of Section II are satisfied. Then two unbounded, densely defined, operators $c$ and $R$ on $\Hil$ exist,  such that $[c,c^\dagger]=\1$ and $R$ is positive, self adjoint and admits an unbounded inverse $R^{-1}$. Moreover
\be
a=RcR^{-1},\qquad b=Rc^\dagger R^{-1},
\label{51}\en
and, introducing $\hat\varphi_n=\frac{{c^\dagger}^n}{\sqrt{n!}}\,\hat\varphi_0$, $c\varphi_0=0$, we have the following: $\hat\varphi_n\in D(R)\cap D(R^{-1})$, for all $n\geq0$, and the sets $\{R\hat\varphi_n\}$ and $\{R^{-1}\hat\varphi_n\}$ are biorthogonal bases of $\Hil$.

Viceversa, let us consider two unbounded, densely defined, operators $c$ and $R$ on $\Hil$ satisfying $[c,c^\dagger]=\1$ with $R$ positive, self-adjoint  with unbounded inverse $R^{-1}$. Suppose that, introduced $\hat\varphi_n$ as above, $\hat\varphi_n\in D(R)\cap D(R^{-1})$, for all $n\geq0$, and that the sets $\{R\hat\varphi_n\}$ and $\{R^{-1}\hat\varphi_n\}$ are biorthogonal bases of $\Hil$. Then  two operators $a$ and $b$ can be introduced for which $[a,b]=\1$, and for which equations (\ref{51}) and Assumptions 1, 2, and 3 (but not 4)  of Section II are satisfied.

\end{thm}

\begin{proof}

To prove the first part of the theorem we recall that the two sets $\F_\varphi=\{\varphi_n, \,n\geq 0\}$ and $\F_\Psi=\{\Psi_n, \,n\geq 0\}$ defined as in Section II are biorthogonal bases of $\Hil$ but they are not Riesz bases.  Hence, defining
\be
S_\varphi\,\Psi_n=\varphi_n,\qquad S_\Psi\,\varphi_n=\Psi_n,
\label{52}\en
for all $n\geq0$, on the domains $D(S_\varphi)=\mbox{linear span } \{\Psi_n\}$ and $D(S_\Psi)=\mbox{linear span } \{\varphi_n\}$, it follows from general results, \cite{you}, that both these operators are unbounded, so that they are not everywhere defined. It is possible to check that $\left<f,S_\varphi\,f\right>\geq0$ for all $f\in D(S_\varphi)$ and $\left<f,S_\Psi\,f\right>\geq0$ for all $f\in D(S_\Psi)$. In particular, if $f\neq0$, both these mean values are strictly positive. It is straightforward to check that, as in the previous section, $S_\varphi=S_\Psi^{-1}$, and that both operators are symmetric:
$$
\left\{
\begin{array}{ll}
 \left<f,S_\varphi\,g\right>=\left<S_\varphi\,f,g\right>, \quad \forall\,f,g\in D(S_\varphi),\\
 \left<f,S_\Psi\,g\right>=\left<S_\Psi\,f,g\right>, \quad \forall\,f,g\in D(S_\Psi).\\
\end{array}
\right.
$$
In these conditions it is known, \cite{ped}, that each one of these operators admits a self-adjoint extension, which is also positive. We call these extensions $\hat S_\varphi$ and $\hat S_\Psi$. Using standard results in functional calculus, we can now define square roots of these operators and the following holds:
$$
\hat S_\varphi=\hat S_\Psi^{-1}, \quad \hat S_\varphi^{1/2}=\hat S_\Psi^{-1/2}, \quad \hat S_\varphi^{-1/2}=\hat S_\Psi^{1/2}.
$$
It is easy to check that, for all $n\geq0$, $\varphi_n\in D(\hat S_\varphi^{-1/2})$, so that $D(\hat S_\Psi)=D(\hat S_\varphi^{-1})\subseteq D(\hat S_\varphi^{-1/2})$. Indeed,  we can check that $\|\hat S_\varphi^{-1/2}\varphi_n\|=1$. This is a particular case of the following more general result:
\be
\left<\hat S_\varphi^{-1/2}\varphi_n,\hat S_\varphi^{-1/2}\varphi_k\right>=\left<\varphi_n,\hat S_\varphi^{-1}\varphi_k\right>=\left<\varphi_n,\hat S_\Psi\varphi_k\right>=\left<\varphi_n,\Psi_k\right>=\delta_{n,k},
\label{53}\en
due to the biorthogonality of $\F_\varphi$ and $\F_\Psi$. This suggests to introduce a third set of vectors of $\Hil$, $\F_{\hat\varphi}=\{\hat\varphi_n:=\hat S_\varphi^{-1/2}\varphi_n,\, n\geq0\}$, which is made of o.n. vectors. As in Section III, defining $\hat\Psi_n=\hat S_\Psi^{-1/2}\Psi_n$, does not produce new vectors; again we get  $\hat\Psi_n=\hat\varphi_n$ $\forall\, n\geq0$. We also deduce that  $D(\hat S_\varphi)\subseteq D(\hat S_\varphi^{1/2})$.

Let us notice that, since $D(\hat S_\Psi)\subseteq D(\hat S_\varphi^{-1/2})\subset \Hil$ and since the closure of $D(\hat S_\Psi)$ returns $\Hil$, $\overline{D(\hat S_\varphi^{-1/2})}^{\,\|\,\|}=\Hil$. Analogously,  $\overline{D(\hat S_\varphi^{1/2})}^{\,\|\,\|}=\Hil$. Moreover, $\forall\,n\geq 0$, $\hat\varphi_n\in D(\hat S_\varphi^{-1/2})\cap D(\hat S_\varphi^{1/2})$: indeed, a straightforward computation shows that $\hat S_\varphi^{1/2}\hat\varphi_n=\varphi_n$ and that $\hat S_\varphi^{-1/2}\hat\varphi_n=\hat S_\varphi^{-1}\varphi_n=\hat S_\Psi\varphi_n=\Psi_n$.

Finally, if $f\in D(\hat S_\varphi^{-1/2})$ is orthogonal to all $\hat\varphi_n$, $f=0$. Hence, due to the density of $D(\hat S_\varphi^{-1/2})$ in $\Hil$, we conclude that $\F_{\hat\varphi}$ is an o.n. basis of $\Hil$, \cite{han}. On $\F_{\hat\varphi}$ we define the {\em standard} annihilation operator $c$ as usual, $c\,\hat\varphi_n=\sqrt{n\,}\hat\varphi_{n-1}$, whose adjoint is the creation operator $c^\dagger\,\hat\varphi_n=\sqrt{n+1\,}\hat\varphi_{n+1}$. We can rewrite the first of these equation as $c\,\hat S_\varphi^{-1/2}\,\varphi_n=\sqrt{n\,}\,\hat S_\varphi^{-1/2}\,\varphi_{n-1}$, which implies, first of all, that $c\,\hat S_\varphi^{-1/2}\,\varphi_n\in D(\hat S_\varphi^{1/2})$. Also, $\hat S_\varphi^{1/2}\,c\,\hat S_\varphi^{-1/2}\,\varphi_n=\sqrt{n\,}\,\varphi_{n-1}$ which, compared with $a\,\varphi_n=\sqrt{n\,}\,\varphi_{n-1}$, shows that $a=\hat S_\varphi^{1/2}\,c\,\hat S_\varphi^{-1/2}$.

In a similar way, $c^\dagger\,\hat\varphi_n=\sqrt{n+1\,}\hat\varphi_{n+1}$ can be rewritten as $c^\dagger\,\hat S_\varphi^{-1/2}\,\varphi_n=\sqrt{n+1\,}\,\hat S_\varphi^{-1/2}\,\varphi_{n-1}$. Therefore $c^\dagger\,\hat S_\varphi^{-1/2}\,\varphi_n\in D(\hat S_\varphi^{1/2})$ and $\hat S_\varphi^{1/2}\,c^\dagger\,\hat S_\varphi^{-1/2}\,\varphi_n=\sqrt{n+1\,}\,\varphi_{n+1}$ which, compared with $b\,\varphi_n=\sqrt{n+1\,}\,\varphi_{n+1}$, shows that $b=\hat S_\varphi^{1/2}\,c^\dagger\,\hat S_\varphi^{-1/2}$. This proves (\ref{51}), identifying $R$ with $\hat S_\varphi^{1/2}$. Also, since $R\hat\varphi_n=\hat S_\varphi^{1/2}\hat\varphi_n=\varphi_n$ and $R^{-1}\hat\varphi_n=\hat S_\varphi^{-1}\varphi_n=\Psi_n$, the linear spans of both $\{R\hat\varphi_n\}$ and $\{R^{-1}\hat\varphi_n\}$ are biorthogonal bases of $\Hil$.

\vspace{3mm}

Let us now prove the inverse statement. Because of our assumptions, the set $\F_{\hat\varphi}$ of vectors $\hat\varphi_n=\frac{{c^\dagger}^n}{\sqrt{n!}}\,\hat\varphi_0$, $c\varphi_0=0$, is an o.n. basis in $\Hil$ and $\hat\varphi_n\in D(R)\cap D(R^{-1})$, $\forall\,n\geq0$. Then we define, for all $n\geq0$, $\varphi_n=R\hat\varphi_n$, $\Psi_n=R^{-1}\hat\varphi_n$, $\F_\varphi=\{\varphi_n,\,n\geq0\}$, $\F_\Psi=\{\Psi_n,\,n\geq0\}$, and $D_\varphi$ and $D_\Psi$ their linear span, which are both dense in $\Hil$ since, by assumption, $\F_\varphi$ and $\F_\Psi$ are (biorthogonal) bases of $\Hil$.

We can now introduce lowering and raising operators on $\F_\varphi$ as in (\ref{45}). In particular, iterating $b\,\varphi_n=\sqrt{n+1}\,\varphi_{n+1}$, we get $\varphi_n=\frac{b^n}{\sqrt{n!\,}}\,\varphi_0$ and we also find that $b^\dagger\Psi_n=\sqrt{n-1}\,\Psi_{n-1}$. The first equation, $a\,\varphi_n=\sqrt{n-1}\,\varphi_{n-1}$, produces $a^\dagger\Psi_n=\sqrt{n+1}\,\Psi_{n+1}$, which, again by iteration, gives $\Psi_n=\frac{{a^\dagger}^n}{\sqrt{n!\,}}\,\Psi_0$.

It is now a simple exercise to check that:

\begin{enumerate}

\item $a\varphi_0=0$ and $\varphi_0\in D^\infty(b)$. Hence Assumption 1 is satisfied.

\item $b^\dagger\Psi_0=0$ and $\Psi_0\in D^\infty(a^\dagger)$. Hence Assumption 2 is satisfied.

\item With similar techniques as in the first part of the proof we deduce that $b=R\,c\,R^{-1}$ and $a=R\,c\,R^{-1}$, which could also be checked  computing directly their action on the vectors $\hat\varphi_n$.

\item $\overline{D_\varphi}^{\,\|\,\|}=\overline{D_\Psi}^{\,\|\,\|}=\Hil$. Hence Assumption 3 is satisfied.

\item since $\F_\varphi$ and $\F_\Psi$ are obtained from the o.n. basis $\F_{\hat\varphi}$ via the action of an unbounded, invertible, operator with unbounded inverse, they cannot be Riesz bases, \cite{you}. Hence Assumption 4 is violated.

\end{enumerate}

This concludes the proof.

\end{proof}

\subsection{Physical examples}

We conclude this section with some examples, arising from quantum mechanics, in which the operators $\hat S_\varphi$ and $\hat S_\Psi$ can be explicitly identified. These examples are reviewed in \cite{bagrev}, where the original references and more examples (even in $d>1$) can be found.

\subsubsection{The extended quantum harmonic oscillator}

The hamiltonian of this model, introduced in \cite{dapro}, is the  non self-adjoint operator $H_\beta=\frac{\beta}{2}\left(p^2+x^2\right)+i\sqrt{2}\,p$, where $\beta$ is a positive parameter and $[x,p]=i$.
Introducing the standard bosonic operators $a=\frac{1}{\sqrt{2}}\left(x+\frac{d}{dx}\right)$, $a^\dagger=\frac{1}{\sqrt{2}}\left(x-\frac{d}{dx}\right)$, $[a,a^\dagger]=\1$, and the number operator $N=a^\dagger a$, we can write $H_\beta=\beta N+(a-a^\dagger)+\frac{\beta}{2}\,\1$ which, introducing further the operators
\be
\hat A_\beta=a-\frac{1}{\beta}, \qquad \hat B_\beta=a^\dagger+\frac{1}{\beta},
\label{56}\en
can be written as
\be
H_\beta=\beta\left(\hat B_\beta \hat A_\beta+\gamma_\beta\,\1\right),
\label{57}
\en
where $\gamma_\beta=\frac{2+\beta^2}{2\beta^2}$. It is clear that, for all $\beta>0$, $\hat A_\beta^\dagger\neq \hat B_\beta$ and that $[\hat A_\beta, \hat B_\beta]=\1$. Hence we have to do with pseudo-bosonic operators which, as proved in \cite{bagpb4}, satisfy Assumptions 1, 2 and 3 but not Assumption 4. Indeed we have deduced that  $\hat S_\varphi=e^{2 (a+a^\dagger)/\beta}$, which is unbounded with unbounded inverse. We have PB which are not regular.

\subsubsection{The Swanson hamiltonian}

The starting point is the following non self-adjoint hamiltonian, \cite{dapro}:
$$
H_\theta=\frac{1}{2}\left(p^2+x^2\right)-\frac{i}{2}\,\tan(2\theta)\left(p^2-x^2\right),
$$
where $\theta$ is a real parameter taking value in $\left(-\frac{\pi}{4},\frac{\pi}{4}\right)\setminus\{0\}=:I$. It is clear that $H_\theta^\dagger\neq H_\theta$, for all $\theta\in I$.  Introducing the annihilation and creation operators $a$ and $a^\dagger$ as usual, we write
$$
H_\theta=N+\frac{i}{2}\,\tan(2\theta)\left(a^2+(a^\dagger)^2\right)+\frac{1}{2}\,\1,
$$
where $N=a^\dagger a$. This hamiltonian can be still rewritten,  by introducing the operators
\be
\left\{
\begin{array}{ll}
A_\theta=\cos(\theta)\,a+i\sin(\theta)\,a^\dagger,  \\
B_\theta=\cos(\theta)\,a^\dagger+i\sin(\theta)\,a,
\end{array}
\right.
\label{41c}\en
as
$$
H_\theta=\omega_\theta\left(B_\theta\,A_\theta+\frac{1}{2}\1\right),
$$
where $\omega_\theta=\frac{1}{\cos(2\theta)}$ is well defined since $\cos(2\theta)\neq0$ for all $\theta\in I$. It is clear that $A_\theta^\dagger\neq B_\theta$ and that $[A_\theta,B_\theta]=\1$. In \cite{bagpb4} we have proven that these operators satisfy Assumptions 1, 2 and 3 but not Assumption 4. In particular we have deduced that  $\hat S_\varphi=|\alpha|^2\,e^{i\theta\,(a^2-{a^\dagger}^2)}$, where $\alpha\in\Bbb{C}$ is arbitrary but fixed. which is unbounded with unbounded inverse. Again, we find PB which are not regular.

\section{Conclusions}

In this paper we have discussed the relation between  RPB and PB with ordinary bosons. As the two theorems proven here clearly show, there is a strong connection between these {\em excitations}, at least under suitable assumptions. Which are the relevant assumptions are clarified by the theorems: for instance, if we just consider operators satisfying $[a,b]=\1$, this is not enough to get any {\em relevant functional structure}. If, as an example, we take $a=\frac{d}{dx}$, $b=x$ and $\Hil=\Lc^2(\Bbb R)$, no square integrable function $\varphi_0(x)$ exists with the required properties. So Assumption 1 (and Assumption 2 as well) is not satisfied. So we cannot introduce, starting from $a$ and $b$, a basis of $\Hil$. This suggests that,  while Assumption 4 can be avoided, and Assumption 3 could be weakened by considering relevant subspaces of $\Hil$, Assumption 1 and 2 are absolutely necessary.

Further analysis on these operators are in progress.

\section*{Acknowledgements}
   The author acknowledge M.I.U.R. for financial support.

\end{document}